\documentclass[11pt]{article} 
\usepackage{etex}
\usepackage{geometry}
\usepackage{float}                
\usepackage{url}
\usepackage{fca}
\usepackage{lineno}

\usepackage{subfig}

\usepackage{xcolor}

\usepackage[utf8]{inputenc}
\usepackage[T1]{fontenc}

\usepackage{tipa}

\usepackage{amsfonts}
\usepackage{amssymb}
\usepackage{amsmath}
\usepackage{amsthm}
\usepackage{multirow}
\usepackage{comment}
\usepackage{caption}
\usepackage{pstricks}
\usepackage[colorlinks]{hyperref}
\usepackage{tikz}
\tikzset{node distance=2cm, auto}

\usepackage{multicol}
\usepackage{amsmath}
\usepackage{graphicx}
\usepackage[numbers]{natbib}
\usepackage{epsfig}
\usepackage{color}
\usepackage{float}
\usepackage[ruled,vlined]{algorithm2e}
\usepackage{graphics}

\newtheorem{definition}{Definition}
\newtheorem{proposition}{Proposition}
\newtheorem{property}{Property}

\newtheorem{theorem}{Theorem}
\newtheorem{lemma}{Lemma}

\newtheorem{corollary}{Corollary}
\newtheorem{example}{Example}
\newtheorem{question}{Question}

\renewenvironment{proof}{{\bf Proof:}}{$\square$}

\def\cF{\mathbb{F}}

\def\cM{\mathbb{M}}

\newcommand{\extM}[1]{\cM_{\overline{#1}}}
\newcommand{\extF}[1]{\cF_{\overline{#1}}}
\newcommand{\extSigma}[1]{\Sigma_{\overline{#1}}}
\newcommand{\extPhi}[1]{\Phi_{\overline{#1}}}

\newcommand{\ie}{\emph{i.e.}\xspace}

\newcommand{\Iff}{{if and only if}\xspace}

\begin{document}

\title{Representations  for  the largest Extension of a closure system}

\author{Karima Ennaoui\thanks{Universit\'e Clermont-Auvergne, France,  LIMOS, UMR 6158 CNRS} \and Khaled Maafa\thanks{Universit\'{e} de Bejaia, Unit\'{e} de Recherche LaMOS. Facult\'{e} des Sciences Exactes,  Alg\'{e}rie.} \and Lhouari Nourine \thanks{Universit\'e Clermont-Auvergne, France,  LIMOS, UMR 6158 CNRS}
} 
    
%
%
%
%
%
%
%
%
%
%
%
%
%

\maketitle

\begin{abstract}

 	We consider extension of a closure system on a finite set $S$ as a closure system on the same set $S$ containing the given one as a sublattice. 
	A closure system can be represented in different ways, e.g.  by an implicational base or by the set of its  meet-irreducible elements.
	When a closure system is described by an implicational base,  we provide a characterization of the implicational base for the largest extension.
	We also show that the largest extension  can be handled by  a small modification of the implicational base of the input closure system. This answers a question asked  in  \cite{Yoshikawa2016}. 
	Second, we are interested in computing the largest  extension when the closure system is given by the set of all its meet-irreducible elements. 
	We give an incremental polynomial time algorithm to compute the largest extension of a closure system, and left open if the number of meet-irreducible elements grows exponentially.  	

\end{abstract}

\paragraph{keywords}
	Closure system, Lattice, Implicational base, Largest extension,  Incremental polynomial enumeration algorithm.

\maketitle

\section{introduction}

An extension (weak-extension)  of a closure system on a finite set $S$ is a closure system on the same set $S$ containing the given one as a sublattice (subset). Extensions of closure systems have been studied by several authors \cite{Adaricheva2004, Adaricheva2003, Ganter2007, Nation2004} and it was established that every closure system has a largest extension.  
Motivated by the computation of the Guigues-Duquenne base (GD-base for short) for a closure system given by its set of meet-irreducible elements, G\'{e}ly and Nourine \cite{Gely2006}  investigate   the family  of all closure systems  on the same set having the same non-unitary implications in their GD-base. This family of closure systems is characterized by the addition  of  unitary  implications to the GD-base. Given the set of meet-irreducible elements of a closure system, they provide a polynomial time algorithm to compute the maximal set of unitary implications that can be added by keeping the same set of non-unitary implications of the GD-base. They also characterize the set of meet-irreducible elements for the  new closure system when such unitary implications are added to its GD-base. Clearly the new closure system has a smaller number of closed sets and meet-irreducible elements than the given one, and thus accelerate existing algorithms (e.g; Next-closure \cite{ganter10}) for computing the GD-base from the set of meet-irreducible elements.
 It is worth noticing that an extension of a closure system can be obtained from  its GD-base by removing one or more  unitary implications from the basis. 
Ganter and Reppe \cite{Ganter2007} show that  an extension of a closure system  can be described by the non-unitary implications. They also characterize the lattice of all extensions of a closure system and show that it corresponds to an interval in the lattice of all closure systems.

Motivated by  join-semidistributive lattices and convex geometry embedding, Adaricheva and Nation \cite{Adaricheva2004, Adaricheva2003} provide a construction which yields the largest  extension. It has been observed in \cite{Yoshikawa2016} that the direct use of the characterization of the largest extension given in \cite{Adaricheva2004} leads to an exponential time algorithm for building the largest extension of a closure system, since one has to check a condition for every element of every subset of the universe. 
Thus, the main motivation of the present paper is to study different representations for the largest extension of a given closure system.
However, computing meet-irreducible sets representation of the largest extension has remained a challenging task.  

\par{} Our contribution in the present paper consists in giving an efficient algorithm for computing a representation of  the largest extension of a closure system.  We consider that the input closure system is either described by any implicational base or its set of meet-irreducible elements.
When a closure system is described by an implicational base,  we give a characterization of an implicational base for  the largest extension. It uses a smaller number of implications than the input. 
This answers a question asked in  \cite{Yoshikawa2016}.

	Second, when the closure system is given by the set of all its meet-irreducible elements,
	we give an incremental polynomial time algorithm to compute the set of meet-irreducible elements of the largest extension of a closure system.  We left open if the number of meet-irreducible elements of the largest extension grows exponentially in the number of meet-irreducible elements of the input closure system.  	

\section{Preliminaries}

The objects considered in this paper are supposed  finite. We refer to \cite{ganterwille99, gratzer11} for more details on posets and lattices.

\paragraph{Closure system} A closure system over a finite set $S$ is a family $\cF$ of subsets over $S$, containing $S$ and closed under intersection. The elements in $\cF$ are called closed sets. 
A closure operator $\Phi$ is a map from and to the powerset of $S$, satisfying that $\Phi$ is extensive ($X\subseteq \Phi(X)$), increasing (if $X\subseteq Y$ then $\Phi(X)\subseteq \Phi(Y)$) and idempotent ($\Phi(\Phi(X))=\Phi(X)$). When $\Phi(X)=X$, then $X$ is called $\Phi$-closed set. The family of $\Phi$-closed sets is a closure system. It is well known that closure operators are in one-to-one correspondence with  closure systems. Moreover, a closure system ordered under set-inclusion $(\cF,\subseteq)$ is a lattice. 

There are numerous ways to represent a closure system such as implicational bases or meet-irreducible sets in $\cF$. 

\paragraph{Implicational base} An implicational base $\Sigma$ over $S$ is defined by a set of implications or rules $L\rightarrow R$ with $(L,R) \in 2^S\times 2^S$. The $\Sigma$-closure of a set $X\subseteq S$ is the smallest set denoted by $X^\Sigma$ containing $X$ and verifying for every $L\rightarrow R\in \Sigma$ that if $L\subseteq X^\Sigma$ then $R\subseteq X^\Sigma$. The set of all $\Sigma$-closed sets form a closure system $\cF$ over $S$. Many equivalent implicational bases yield to a same closure system. For example, the Guigues-Duquenne implicational base (GD-base for short) \cite{gd86} is among implicational bases that contain minimum number of implications. Given an implicational base $\Sigma$, we will identify the closure operator $\Phi$ with the $\Sigma$-closure of $X$, i.e. $\Phi(X)=X^\Sigma$.

\paragraph{Meet-irreducible sets} The meet-irreducible sets of a closure system $\cF$ is the smallest subset $\cM\subseteq \cF$, such that any closed set $F\not =S$ of $\cF$  is the intersection of a subset of $\cM$. In other words, a set $M$ is a meet-irreducible of $\cF$ \Iff for every   of $X$ and $Y$ in $\cF$ we have:  $M=X\cap Y$ implies that   $M=X$ or $M=Y$. 

\paragraph{Covering relation}  Let   $F$ and $F'$ in $\cF$ such that $F\subset F'$. We say that $F'$ covers $F$ if for any $F''\in \cF$ with $F\subset F''\subseteq F'$ then  $F'=F''$. Note that a closed set is  a meet-irreducible set in $\cF$ \Iff it has  a unique cover in $\cF$ denoted by $F^*$.

\paragraph{Closure system Extension} The extension of a closure system $\cF$ is a closure system containing $\cF$ as a subset. 
Among all possible extensions of a closure system, we will distinguish those that preserve the sublattice property (the closure  of any  two closed sets in $\cF$  remains unchanged in the extension). A reader can be referred to \cite{Adaricheva2004, Adaricheva2003} for other properties that can be preserved by extentions such as join-semidistributivity or convex geometry.

\begin{definition}
	Let $\cF$, $\cF'$ be two closure systems and their respectively corresponding closure operators $\Phi$, $\Phi'$. We say that $\cF'$ is an \textbf{extension} of $\cF$ if $\cF \subseteq \cF'$ and for every $F_1$ and $F_2$ in $\cF$, $\Phi(F_1\cup F_2)=\Phi'(F_1\cup F_2)$.
\end{definition}

Let $\mathcal{C}(S)$ be the set of all closure systems over a finite set $S$. It is well known that $(\mathcal{C}(S),\subseteq)$ is a lattice \cite{Caspard2003241}. The set of all extensions of a closure system $\cF$ is an interval in $\mathcal{C}(S)$ \cite{Adaricheva2004}. So, every closure system $\cF$ over $S$ has a unique largest extension $\cF_{max}$. 
This largest extension has been characterized in \cite{Adaricheva2004} by the following property:
$F\subseteq S$ is in $\cF_{max}$ if and only if for any $A \subseteq S$ we have:
\begin{equation*}
( \forall a \in A,~~\Phi_{max}(a)\subseteq F)\Rightarrow \Phi_{max}(A)\subseteq F.
\end{equation*}

Note that the verification of this property can take an exponential time, since it invokes every subset $A \subseteq S$.

\paragraph{Notation} We consider in the rest of this paper:

\begin{itemize}

\item A closure system $\cF$ over a finite set $S$, either described by its implicational base $\Sigma$ or its meet-irreducible sets $\cM\subseteq \cF$. In both cases, we use a closure operator $\Phi$ corresponding to the closure system $\cF$ defined for any $X\subseteq S$ by: either $\Phi(X)=X^\Sigma$ if $\Sigma$ is given, or $\Phi(X)=\bigcap\{M\in \cM \mid X\subseteq M\}$ when $\cM$ is given.

\item We denote by $U_\Phi$ the subset $\{x\in S \mid  \Phi(x)\neq \{x\}\}$. For each $x\in U_\Phi$, $\Phi_*(x)$ denotes the set $\Phi(x)\setminus\{x\}$. It is worth noticing that $\Phi_*(x)$ is not closed, whenever there are $x,y\in S$ with $\Phi(x)=\Phi(y)$.
	
\item $\Sigma_u=\{x\rightarrow \Phi(x), x \in U_\Phi\}$ is a subset of $\Sigma$, and each implication in $\Sigma_u$ is called unitary. Let $\Sigma_{nu}=\Sigma\setminus \Sigma_u$ the set of non-unitary implications, verifying each $|L|>1$. We denote $\cF_u$ (respectively $\cF_{nu}$) and $\Phi_u$ (respectively $\Phi_{nu}$) the closure system and closure operator corresponding to $\Sigma_u$ (respectively $\Sigma_{nu}$). 
\end{itemize}

For the sake of readability, when it is not confusing we  use $x$ instead of $\{x\}$, for example  $\Phi(x)$ instead of  $\Phi(\{x\})$. 

\section{Implicational base of the largest extension}

Given a closure system $\cF$ by an implicational base $\Sigma$ over a finite set $S$. 
We give a necessary and sufficient condition for the implicational base of the largest extension of $\cF$. 

The following lemma shows that any closure system not containing the empty set cannot correspond to a largest extension. In other words, the largest extension must contain the empty set.

\begin{lemma} \label{lemvide}Any implicational base $\Sigma$ for the largest extension cannot contain an implication $L \rightarrow R$ with $L=\emptyset$. 
\end{lemma}
\begin{proof}
Let $\Sigma$ be an implicational base of the largest extension of a closure system.  Suppose that $\Sigma$ contains an implication $L \rightarrow R$ with $L=\emptyset$.  We show that $\cF_{\Sigma}$ is not  the largest extension. Then any closed set contains $\Phi(\emptyset)\not = \emptyset$ whenever $R\not = \emptyset$.
Thus adding the empty set is still an extension which contradicts the fact that is the largest.
\end{proof}

\medskip

Now we show that the largest extension is atomistic, i.e. it cannot contain a unitary implication.

\begin{lemma} \label{lemsingleton}Any implicational base $\Sigma$ for the largest extension cannot contain a unitary implication $x\rightarrow R$ for some $x\in S$. 
\end{lemma}
\begin{proof} Indeed, adding the set $\{x\}$ to $\cF$ is still an extension, since $\{x\}$ cannot be the closure of two closed sets.
\end{proof}

\medskip

Thus, to obtain the largest extension  we eliminate all unitary implications from the original closure system's base $\Sigma$, with $\Sigma$ respecting certain conditions which necessity is illustrated in the following example.

First, consider $\Sigma^1=\{ab \rightarrow c, c\rightarrow d\}$. Then $\Sigma^1_{nu}=\{ab \rightarrow c\}$ and $\Sigma^1_{u}=\{c\rightarrow d\}$. The set $\{a,b,c,d\}=\Phi^1(\{a\} \cup \{b\})\not = \Phi^1_{nu}(\{a\} \cup \{b\})=\{a,b\}^{\Sigma^1_{nu}}= \{a,b,c\}$. So $\cF^1_{nu}$ is not an extension of $\cF^1$, the closure of the closed sets $\{a\}$ and $\{b\}$ has changed, i.e. the lattice $(\cF,\subseteq)$ is not a sublattice of $(\cF^1,\subseteq)$.

Second, let $\Sigma^2=\{ab \rightarrow c, b\rightarrow d\}$. We have $\Sigma^2_{nu}=\{ab \rightarrow c\}$ and $\Sigma^2_{u}=\{b \rightarrow d\}$. However, $\cF^2_{nu}\cup\{a,b\}$ is a strict extension of $\cF^2_{nu}$ and $\cF^2$, and therefore $\cF^2_{nu}$ is not the largest extension of $\cF^2$. 

In theorem \ref{theo:extensionBase}, we give a necessary and sufficient condition for implicational base that represents the largest extension.  An implicational base $\Sigma$ is said {\em ideal-closed} if for any $L\rightarrow R\in \Sigma$, $\Phi_{nu}(L)$ is $\Sigma_u$-closed.

\begin{theorem} \label{theo:extensionBase}
	Let $\Sigma$ be an implicational base for a closure system $\cF$. Then $\Sigma_{nu}$ is an implicational base of the largest extension $\cF_{max}$ if and only if
	 $\Sigma$ is ideal-closed and for every $L\rightarrow R$ and  $x\in L$, $\Phi_*(x)\subseteq \Phi_{nu}(L\setminus\{x\})$.

\end{theorem}
\begin{proof}

First, suppose that $\Sigma_{nu}$ is an implicational base of the largest extension $\cF_{max}$ and does not satisfy one of the conditions of the theorem. 

We distinguish two cases:
\begin{enumerate}

\item $\Sigma_{nu}$ is not ideal-closed. Then, there exists $L\rightarrow R$ in $\Sigma_{nu}$ such that $\Phi_{nu}(L)$ is not $\Sigma_u$-closed. We suppose that $L\rightarrow R$ is not trivial\footnote{Removing trivial implication does not change the closure system}, i.e. $R\not \subseteq L$, otherwise it can be dropped from $\Sigma_{nu}$. 
We distinguish two cases: (a) There are two maximal closed sets $F_1\not = F_2$ with $F_1\cup F_2\subseteq \Phi_{nu}(L)$, then $\Phi_{nu}(F_1\cup F_2)\not = \Phi(F_1\cup F_2)$ since $\Phi_{nu}(L)$ is not $\Sigma_u$-closed. This contradicts that $\cF_{nu}$ is an extension of $F$; (b) There exists a unique maximal closed set $F_L\in \cF$ that verifies $F_L\subseteq \Phi_{nu}(L)$. Then, for two maximal closed set $F_1\not = F_2$ with $F_1\cup F_2\subseteq L$, we have $F_1\cup F_2\subseteq F_L$ and thus $\Phi(F_1\cup F_2)\subseteq F_L$. Hence, adding $L$ to $\cF_{nu}$ is still an extension since $F_L\not = L$. 
This contradicts that $\cF_{nu}$ is the largest extension of $F$.

\item Exists $L \rightarrow R\in \Sigma_{nu}$ and $x\in L$ verifying $\Phi_*(x)\not\subseteq \Phi_{nu}(L\setminus\{c\})$. This implies that $\Phi(x)\neq\{x\}$. Let $K=\{x\}\cup\Phi_{nu}(L\setminus\{x\})$, we prove that $K\notin \cF_{nu}$ and yet $K\in \cF_{max}$. 

First, note that since $\Sigma$ is ideal closed and $x\in L$, then $\Phi_*(x)\subseteq\Phi_{nu}(L)$. And since $L\subseteq K$ then $\Phi_*(x)\subseteq\Phi_{nu}(L)\subseteq\Phi_{nu}(K)$. However, the hypothesis assumes that $\Phi_*(x)\not\subseteq K$. Therefore $K\neq \Phi_{nu}(K)$, equivalently $K$ is not in $\cF_{nu}$.

Second, for every $F_1$ and $F_2$ in $\cF$, $(F_1\cup F_2)\subseteq K$ implies that $(F_1\cup F_2) \subseteq \Phi_{nu}(L\setminus\{x\})$ or $x\in(F_1\cup F_2)$. However, since $\Phi(x)\not\subseteq K$ then $\Phi(x)\not \subseteq  (F_1\cup F_2)$, which implies $x\notin (F_1\cup F_2)$. We conclude that for every $F_1$ and $F_2$ in $\cF$, $(F_1\cup F_2)\subseteq K$ implies that $(F_1\cup F_2) \subseteq \Phi_{nu}(L\setminus\{x\})$. We recall that $\{x\}$ is not in $\cF$. Hence, adding $K$ to $\cF_{nu}$ is still an extension of $\cF$. This contradicts that $\cF_{nu}$ is the largest extension of $F$.

\end{enumerate}

Conversely, suppose that $\Sigma_{nu}$ is ideal-closed, where every $L \rightarrow R$ verifies   $\Phi_*(x)\subseteq \Phi_{nu}(L\setminus\{x\})$ for each  $x\in L$. 
We prove that $\cF_{nu}$ is the largest extension of $\cF$. First, we prove that $\cF_{nu}$ is an extension of $\cF$ and then it is the largest one.

First, we show that $\cF_{nu}$ is an extension of $\cF$:
	\begin{enumerate}
		 \item $\cF \subseteq \cF_{nu}$; Every closed set $F\in\cF$ is also $\Sigma_{nu}$-closed, and therefore $F\in \cF_{nu}$.
		 \item For every $X\not =Y$ in $\cF$, $\Phi(X\cup Y)=\Phi_{nu}(X\cup Y)$.\\
		Assume $\Sigma_{nu}=\{L_1\rightarrow R_1,...,L_n\rightarrow R_n\}$, and recall the set closure algorithm of a subset $Z\subseteq S$ by an implicational base. There exists a sequence of $k$ implications \\
		$$Z=Z_0\xrightarrow[]{L_1\rightarrow R_1} Z_1=Z_0\cup R_1 \xrightarrow[]{}...\xrightarrow[]{}Z_{k-1}\xrightarrow[]{L_k\rightarrow R_k} Z_k=Z_{k-1}\cup R_k=\Phi_{nu}(Z)$$
Obtained by iteratively applying the closure of an implication $L_i\rightarrow R_i$ in to the current set $Z_{i-1}$, until it is closed. This implies that $\Phi_{nu}(X\cup Y)=(X\cup Y)\cup (\bigcup_{1\leq i\leq k} R_i)$. 

We prove that $\Phi_{nu}(X\cup Y)$ is $\Sigma_u$-closed, i.e. for any $a\in \Phi_{nu}(X\cup Y)$, $\Phi(a)\subseteq \Phi_{nu}(X\cup Y)$. We have two cases: (1) $x\in X\cup Y$ and then $\Phi(a)\subseteq X\cup Y$ since $X$ and $Y$ are $\Sigma_u$-closed, (2) $x\in \Phi_{nu}(X\cup Y)\setminus (X\cup Y)$. Then exists $x\in R_i$ for some implication $L_i\rightarrow R_i\in \Sigma_{nu}$. Since $\Sigma_{nu}$ is ideal-closed, we have $\Phi(a)\subseteq \Phi_{nu}(L_i)$ and then $\Phi(a)\subseteq \Phi_{nu}(X\cup Y)$.

	Since $\Phi_{nu}(X\cup Y)$ is $\Sigma_{nu}$-closed, 	we conclude that $\Phi_{nu}(X\cup Y)$ is $\Sigma$-closed. Hence, $\Phi(X\cup Y)=\Phi_{nu}(X\cup Y)$.
	 \end{enumerate}

Second, $\cF_{nu}$ is the largest extension of $\cF$, i.e. $\cF_{max}\subseteq\cF_{nu}$:

Suppose that $\cF_{max}\setminus\cF_{nu}\neq\emptyset$. Let F be a minimal element in $\cF_{max}\setminus\cF_{nu}$. 
Since $F\notin \cF_{nu}$, then the implications subset $\Sigma_{nu}=\{L'\rightarrow R'\in \Sigma_{nu} | L'\subseteq F$ and $R'\not\subseteq F\}$ is not empty. Let $L\rightarrow R$ be an implication in $\Sigma_{nu}$ with $L$ is minimal, \ie for every $L'\rightarrow R'$ in $\Sigma_{nu}$, $L'\not\subset L$. We distinguish two cases:
\begin{itemize}
	\item $F\not\subseteq \Phi(L)$: We have $F$ and $\Phi(L)$ are both in $\cF_{max}$, then $K=F\cap\Phi(L)$ is in $\cF_{max}$. However $L\subseteq K$ and $R\not\subseteq K$. Hence, $K$ is in $\cF_{max}\setminus\cF_{nu}$ and $K\subset F$. This contradicts the fact that F is minimal in $\cF_{max}\setminus\cF_{nu}$.
	
	\item $F\subseteq \Phi(L)$: we distinguish two cases (1) $\Phi_u(L) \subseteq F \subset \Phi(L)$, in this case we have $\Phi_u(L)=\bigcup_{x\in L} \Phi(x)$ and $\Phi_{max}(\bigcup_{x\in L} \Phi(x))\subseteq F\subset \Phi(\bigcup_{x\in L}=\Phi(L))$. Hence, $\Phi_{max}(\bigcup_{x\in L} \Phi(x))\neq \Phi(\bigcup_{x\in L} \Phi(x))$. This contradicts the fact that $\cF_{max}$ is an extension of $\cF$. (2) $L\subseteq F \subset \Phi_u(L)$. Since $\Phi_u(L)\not\subseteq F$, then there exists $x\in L$ with $\Phi_*(x)\not\subseteq F$. We recall that for every $x\in L$, $\Phi_*(x)\subseteq \Phi_{nu}(L\setminus \{x\})$. This means that there exists $L'\rightarrow R'$ in $\Sigma_{nu}$ with $L'\subset (L\setminus\{x\})\subset F$ and $R'\not\subseteq F$. This contradicts the fact that $L$ is minimal in $\Sigma_{nu}$.
\end{itemize}

We conclude that $\cF_{max}=\cF_{nu}$.
\end{proof}

\medskip

Based on theorem \ref{theo:extensionBase}, we deduce that recognizing if $\Sigma_{nu}$ corresponds to the largest extension can be done in polynomial time.  

\begin{corollary}
Given an implicational base $\Sigma$, we can check in polynomial time if $\Sigma_{nu}$ corresponds to the largest extension.
\end{corollary}
\begin{proof} It suffices to check conditions in theorem \ref{theo:extensionBase}, where computing the closure of a set can be done in polynomial time \cite{maier83}.
\end{proof}

\section{Approach for computing the largest extension}

Recall that the implicational base of the largest extension does not contain  implications with a premise empty or singleton, see lemmas \ref{lemvide} and  \ref{lemsingleton}. 
Theorem \ref{theo:extension} describes each iteration, where we drop a unitary implication from $\Sigma$ and generate the new added closed sets. Compliantly to theorem \ref{theo:extensionBase}, we suppose that $\Sigma_{nu}$ is ideal-closed and for every implication $L \rightarrow R$ in $\Sigma_{nu}$ and for every $x\in L$, $\Phi_*(x)\subseteq \Phi_{nu}(L\setminus\{x\})$.

Let $x$ be an element in $U_\Phi$, we note in the following (see figure \ref{fig:extensionClosedSets}): $\Delta_x(\cF)=\{F\in\cF | \Phi_*(x)\nsubseteq F\}$,
 $\Delta^c_x(\cF)= \{F\cup\{x\} | F\in \Delta_x(\cF)\}$,
$\cF_{\overline{x}}=\cF\cup\Delta_x(\cF)$ and 
$\extSigma{x} = \Sigma\setminus \{x\rightarrow\Phi(x)\}$ and denote $\extPhi{x}$ its closure operator.

\begin{figure}[H] 
	\begin{center} 
		\includegraphics[scale=0.6]{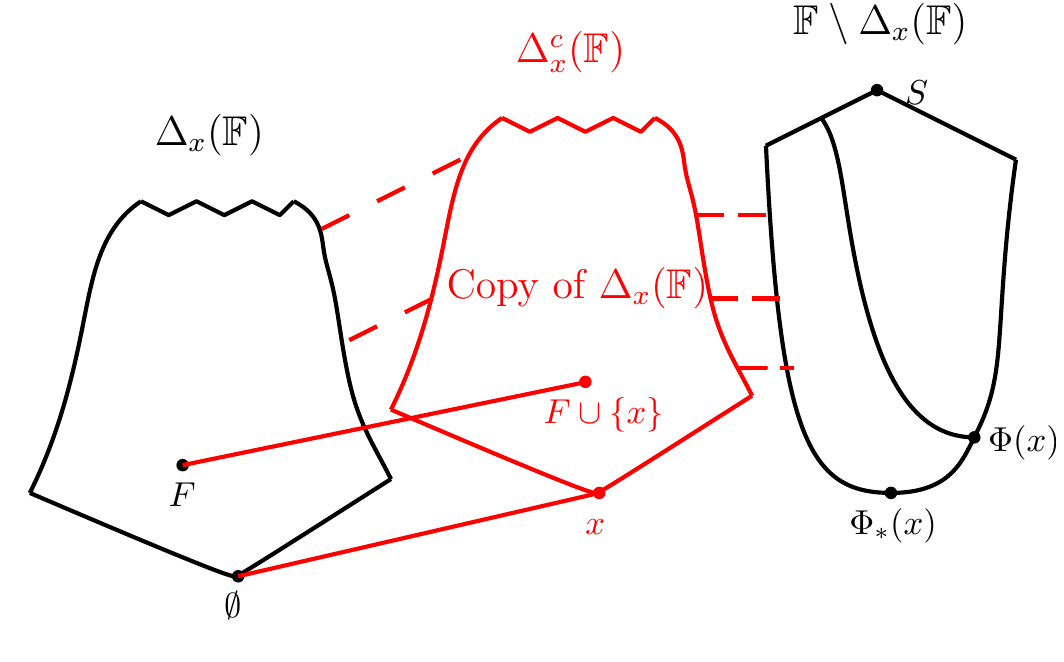} 
		\caption{\small \sl Decomposition of closed sets of $\cF_{\overline{x}}$
			\label{fig:extensionClosedSets}} 
	\end{center} 
\end{figure}

\vspace{-0.6cm}

We prove in  theorem \ref{theo:extension} that $\cF_{\overline{x}}$ is a closure system, an extension of $\cF$ and  $\extSigma{x}$ is its implicational base. 

\begin{theorem}\label{theo:extension}
Let $\Sigma$ be an implicational base satisfying conditions of theorem \ref{theo:extensionBase}, and  $x$ be an element in $U_\Phi$.
	Then:
	\begin{enumerate}
		\item $\cF_{\overline{x}}=\cF\cup\Delta^c_x(\cF)$ is a closure system. 
		\item $\cF_{\overline{x}}$ is an extension of $\cF$. 
		\item $\extSigma{x}$ is an implicational base of $\extF{x}$.		
	\end{enumerate}
\end{theorem}

\begin{proof}
\begin{enumerate}
\item Let $F_1$ and $F_2$ be two sets in $\extF{x}$, we prove that $F_1\cap F_2$ is in $\extF{x}$. If $F_1$ and $F_2$ in $\cF$, then $F_1\cap F_2$ is in $\cF\subseteq \extF{x}$. Now suppose $F_1, F_2$ are in $ \Delta^c_x(\cF)$. Then $F_1\setminus \{x\}$ and $ F_2\setminus \{x\}$ are in $\cF$. Since the intersection $F_1\setminus \{x\}$ and $ F_2\setminus \{x\}$ does not contain $\Phi_*(x)$, then $F_1\cap F_2$ is in $\Delta^c_x(\cF)\subseteq \extF{x}$. Finally, suppose $F_1\in \cF$ and $F_2\in \Delta^c_x(\cF)$. We have two cases :

\begin{itemize}
		\item $x\notin F_1$ then $F_1\cap F_2=F_1\cap (F_2\setminus \{x\})$ is in $\cF$.
		\item $x\in F_1$ then $F_1\cap (F_2\setminus \{x\})$ is in $\cF$. Moreover $F_1\cap (F_2\setminus \{x\})$ does not contain $\Phi_*(x)$, and hence $F_1\cap F_2\in \Delta^c_x(\cF)$.
\end{itemize}
	\item We prove that the closure of $F_1\cup F_2$ is preserved for $F_1$ and $F_2$ in $\cF$. Let $F$ be the closure  of $F_1\cup F_2$ in $\extF{x}$. 
	If $\Phi_*(x)$ is a subset of $F_1$ or $F_2$ then $F$ contains $\Phi_*(x)$ and thus by definition of $\extF{x}$  $F$ belongs to $\cF$. If $\Phi_*(x)$ is not a subset of $F_1$ and $F_2$ then $F_1$ and $F_2$ do not contain $x$. Suppose that $F$ is in $\Delta^c_x(\cF)$. By definition $\Delta^c_x(\cF)$, $F\setminus \{x\}\in \cF$ which is a contradiction, since $F_1\cup F_2 \subseteq F\setminus \{x\}\in \cF$. 

	\item We prove that $\extSigma{x}=\Sigma\setminus \{x\rightarrow\Phi(x)\}$ is the implicational base of $\extF{x}$, i.e. for every $F\subseteq S$, $F\in \extF{x}$ \Iff $F$ is $\extSigma{x}$-closed. 
	
	First, let $F$ be in $\extF{x}$. If $F\in \cF$ then $F$ is $\Sigma$-closed and then $\extSigma{x}$-closed. Now suppose that $F\in \Delta^c_x(\cF)$. Thus $F'=(F\setminus x) \in \cF$ and $F'$ is $\extSigma{x}$-closed. Now let $L\rightarrow R\in \extSigma{x}$ with $L\subseteq F$. We distinguish two cases:

		\begin{itemize}

			\item $x\in L$: since $L\rightarrow R$ is in $\extSigma{x}=\Sigma\setminus \{x\rightarrow\Phi(x)\}$ and $x\in L$, then $L\neq \{x\}$. Therefore $L\rightarrow R$ is non-unitary implication that contains $x$. We recall that $\Sigma_{nu}\subseteq\extSigma{x}$ is ideal closed. Thus, $\Phi_*(x)$ is a subset of $L^{\extSigma{x}}\subseteq F$, which contradicts the fact that $\Phi_*(x)\nsubseteq F$ by definition of $\Delta^c_x(\cF)$. 
			\item $x\notin L$: then we have $L\subseteq F'$, and since $F'$ is $\Sigma$-closed then $R$ is a subset of $F'$ and therefore of $R\subseteq F$. 
		\end{itemize}

	Second, suppose that $F$ is $\extSigma{x}$-closed. We prove that $F$ is either in $\cF$ or in $\Delta^c_x(\cF)$. We distinguish the following cases:
	\begin{itemize}
		\item $\Phi_*(x)\subseteq F$: Then $F$ satisfies $x\rightarrow \Phi(x)$, and thus $F$ is $\Sigma$-closed.
		\item $\Phi_*(x)\not\subseteq F$: We consider $F'=F\setminus\{x\}$. Let $L\rightarrow R$ in $\extSigma{x}$ with $L\subseteq F'\subseteq F$. Then $\extPhi{x}(L)\subseteq F$. Note that if $x\in L^{\extSigma{x}}$, then $\Phi_*(x)\subseteq \extPhi{x}(L)\subseteq F$ (because $\Sigma_{nu}\subseteq\extSigma{x}$), and since $\Phi_*(x)\not\subseteq F$ then $x\notin \extPhi{x}(L)$. We conclude that $R\subseteq \extPhi{x}(L)\subseteq F'$. Hence, $F'$ is ${\extSigma{x}}$-closed and since $x\notin F'$ then $F'$ is ${\Sigma}$-closed. Therefore, $F$ is in $\Delta^c_x(\cF)$. 
		
		\end{itemize}
	\end{enumerate}
We conclude that $F$ is in $\extF{x}$ \Iff it is $\extSigma{x}$-closed.
\end{proof}

\medskip

In \cite{Ganter2007}, it is shown that the lattice of all extensions is an interval of the lattice of all closure systems, and thus it is a convex geometry. It follows that any permutation of $U_\Phi=\{x_1,...,x_k\}$ corresponds to a path from $\cF$ to $\cF_{max}$. This justifies that the order of removing the unitary implications is an invariant.

Let $\Phi$ be a closure operator over a finite set $S$, $\cF$ be its corresponding closure system and $U_\Phi=\{x_1,...,x_k\}$. We define the recursively composed closure operator $\extPhi{i}$ corresponding to $\extF{i}$ after removing the $i^{th}$ first unitary implications $\{x_j\rightarrow\Phi(x_j)| 1\leq j\leq i\}$. Using this notation $\extPhi{0}$ corresponds to $\extF{0}=\cF$ and $\extPhi{k}$ corresponds to $\extF{k}=\cF_{max}$.

\begin{corollary}\label{theo:maxExtension}
	Let $\Phi$ be a closure operator over a finite set $S$, $\cF$ be its corresponding closure system and $U_\Phi=\{x_1,...,x_k\}$. Then $\extF{k}$ is the largest extension of $\cF$.
\end{corollary}
\begin{proof} 
	By theorem \ref{theo:extension}, we deduce that $\extF{k}$ is an extension of $\cF$, and theorem \ref{theo:extensionBase} proves that removing all unitary implications yields a maximal extension. Thus $\extF{k}=\cF_{max}$.
\end{proof}

\medskip

The following example illustrates the strategy for building the largest extension of a closure system.

\begin{example}\label{fig:maxExtensionGenerationExpleTxt}
Consider  the closure system $\cF=\{\emptyset, a,b,ac,ad,abcd\}$ pictured in  Figure \ref{fig:maxExtensionGenerationExple}(a) as a lattice. Its implicational base is $\Sigma=\{c\rightarrow a, d\rightarrow a, ab\rightarrow cd, acd\rightarrow b \}$ and $U_\Phi=\{c,d\}$. Figure \ref{fig:maxExtensionGenerationExple}(b) shows the extension $\extF{c}$ when removing $c\rightarrow a$ and  Figure \ref{fig:maxExtensionGenerationExple}(b) shows the largest extension of $\cF$ when removing both $c\rightarrow a$ and $d\rightarrow a$. We can verify that  its implicational  base is $\Sigma=\{ab\rightarrow cd, acd\rightarrow b \}$.

\begin{figure}[h]
\begin{center}
\includegraphics[width=0.2\textwidth]{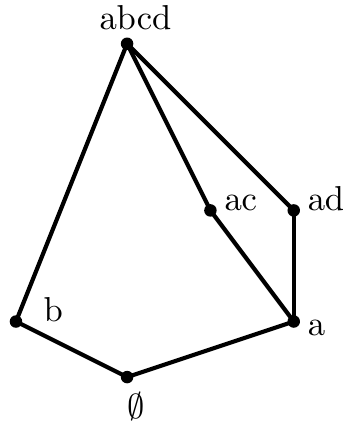}
\hspace{1cm}
\includegraphics[width=0.2\textwidth]{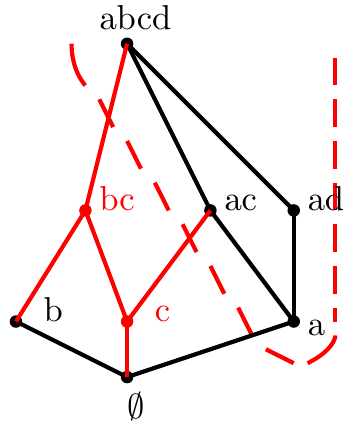}
\hspace{1cm}
\includegraphics[width=0.35\textwidth]{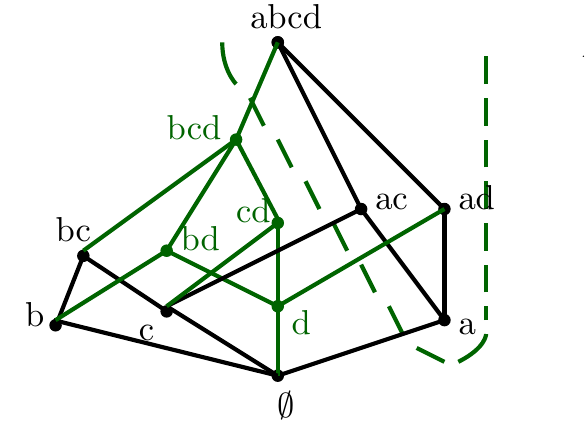}\\
\hspace{-1cm}(a) \hspace{4cm}(b)\hspace{4cm}(c) \hspace{4cm}\\

\end{center}
\caption{\small \sl Successive applications of theorem \ref{theo:extension} to compute the  largest extension}
	\label{fig:maxExtensionGenerationExple}
\end{figure}

\end{example}

\section{Meet irreducible sets of the largest extension}

According to theorem \ref{theo:extension}, given $x\in U_\Phi$ then any closed set $F\in \Delta_x(\cF)$ yields a new element $F\cup \{x\}\in \Delta^c_x(\cF)$. Thus meet-irreducible elements of $\extF{x}$ are either in $\cM$ or in $\Delta^c_x(\cF)$. Since $\Delta^c_x(\cF)$ is a copy of $\Delta_x(\cF)$, then any meet-irreducible  $F$ in $ \Delta^c_x(\cF)$, $F\setminus \{x\}$ has at most one cover in $\Delta_x(\cF)$. Thus, to obtain meet-irreducible sets in $\Delta^c_x(\cF)$, we need to locate closed sets that have at most one cover in $\Delta_x(\cF)$.

 In this section we will give a characterization of meet-irreducible elements in $\cM$ that remain meet-irreducible in $\extF{x}$, and the new meet-irreducible elements in $\Delta^c_x(\cF)$. 
Given $x\in U_\Phi$, we define a partition of $\cM_\cF$ (see figure \ref{fig:extensionMeetsDecomposition} where $\cF$ is pictured as a lattice) as follows. 
\begin{itemize}
	\item $\cM^1$ are meet-irreducible sets that contain $\Phi(x)$;
	\item $\cM^2$ are meet-irreducible sets that contain $\Phi_*(x)$ but not $x$;
	\item $\cM^3$ are meet-irreducible sets that are covered by closed sets containing $x$, i.e. $M\in \cM^3$ \Iff for any $a\in S\setminus M,$ $ \Phi(M\cup{a})$ contains $x$ and thus $\Phi(x)$.
	\item $\cM^4$ are meet-irreducible sets that are covered by closed sets containing $\Phi_*(x)$ but not $x$, i.e. $M\in \cM^4$ \Iff exists $a\in S\setminus M, $ $ \Phi(M\cup{a})$ contains $\Phi_*(x)$ but not $x$. 
	
	\item $\cM^5$ are meet-irreducible sets that do not contain $\Phi_*(x)$ and are not in $\cM^3\cup \cM^4$.
\end{itemize}

\begin{figure}[H] 
	\begin{center} 
		\includegraphics[scale=1]{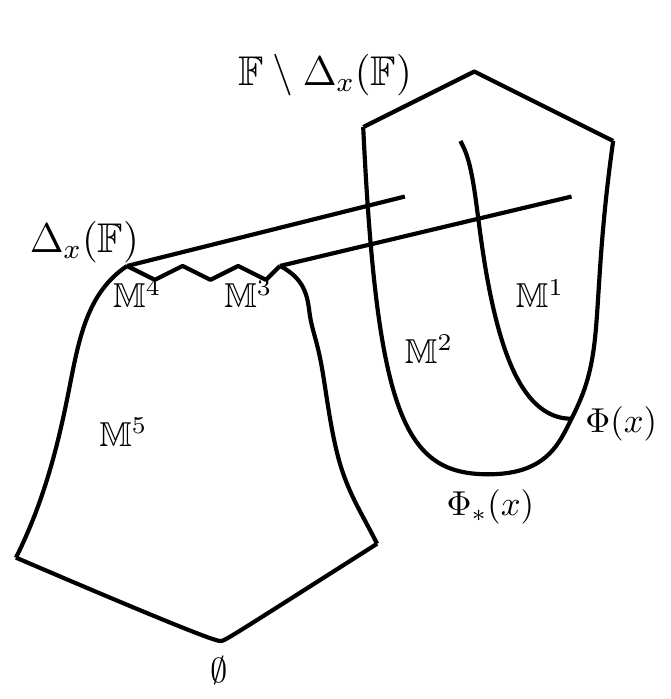} 
		\caption{\small \sl Partition of meet-irreducible sets of $\cF$ relative to $x\in U_\Phi$
			\label{fig:extensionMeetsDecomposition}} 
	\end{center} 
\end{figure}

\begin{example} 
The meet-irreducible sets of the closure system in  figure \ref{fig:maxExtensionGenerationExple}(a) are $\cM=\{\{b\},$ $\{a,c\},\{a,d\}\}$. The Partition relative to $c$ of $\cM$ is $\cM^1=\{\{a,c\}\}$, $\cM^2=\{\{a,d\}\}$, $\cM^3=\{\{b\}\}$ and $\cM^4=\cM^5=\emptyset$. 
\end{example}

\begin{corollary} Meet-irreducible elements of $\extF{x}$ are either meet-irreducible in $\cF$ or are in $\Delta^c_x(\cF)$.
\end{corollary}
\begin{proof} If a closed set $F\in \cF$ has two covers in $\cF$ then it has at least two covers in $\extF{x}$.
\end{proof}

\medskip

The following properties describe meet-irreducible sets of $\cF$ that remain meet-irreducible in the extension $\extF{x}$. 
\begin{proposition}\label{prop:extensionMeets1}
	Let $F\in \cM_\cF$. Then $F$ is a meet-irreducible in $\extF{x}$ \Iff $F\in (\cM^1\cup\cM^2\cup\cM^3)$.

\end{proposition}
\begin{proof}
Let $F\in \cM_\cF$ be a meet-irreducible in $\extF{x}$, and suppose $F\notin (\cM^1\cup\cM^2\cup\cM^3)$. Then $F\in (\cM^4\cup\cM^5)$. This means that the unique cover of $F$ in $\cF$ does not contain $x$. Moreover $F\cup \{x\}$ is also a cover of $F$ in $\extF{x}$, which means that $F$ has at least two covers in $\extF{x}$. This is a contradiction.

Conversely suppose $F\in (\cM^1\cup\cM^2\cup\cM^3)$. We show that $F$ has a unique cover in $\extF{x}$ and thus remains meet-irreducible.
	If $F$ is in $\cM^1\cup\cM^2$, it implies that $F$ contains $\Phi_*(x)$ and thus is still having one cover in $\extF{x}$. If $F\in \cM^3$ then $ \Phi(M\cup{x})$ is the unique cover of $F$ and it contains $\Phi(x)$. Thus $F\cup\{x\} \in \Delta^c_x(F)$ and $F\subseteq F\cup \{x\} \subseteq \Phi(M\cup\{x\})$. Thus $\Phi(M\cup\{x\})$ is the unique cover of $F\cup \{x\}$ in $\extF{x}$. Since $F$ is a meet-irreducible in $\cF$ then $F\cup \{x\}$ is the unique cover of $F$ in $\extF{x}$.
\end{proof}

\medskip

 Notice that if $F\in\Delta^c_x(\cF)$ is a new meet-irreducible in $\extF{x}$, then $F\setminus\{x\}$ must have at most one cover in $\Delta_x(\cF)$. There are two kinds for such meet-irreducible sets, those for which $F\setminus\{x\}$ is a meet-irreducible in $\cF$ and those for which $F\setminus\{x\}$ is not, denoted by  $\cM^6$.

$\cM^6=\{F\in \Delta_x(\cF) \mid F\notin \cM_{\cF}, F'$ is the unique cover of $F$ in $\Delta_x(\cF), \Phi(F\cup \{x\})=\Phi(F'\cup \{x\})\}$.

\begin{example} 
The Partition relative to $d$ of $\cM_c$ of the closure system in figure \ref{fig:maxExtensionGenerationExple}(b) is $\cM^1=\{\{a,d\}\}$, $\cM^2=\{\{a,c\}\}$, $\cM^3=\{\{b,c\}\}$, $\cM^4=\emptyset$, $\cM^5=\{\{b\}\}$ and  $\cM^6=\{\{c\}\}$. Thus $\{c,d\}$ is a new meet-irreducible in $\Delta^c_d(\extF{c})$ as shown in figure \ref{fig:maxExtensionGenerationExple}(c).
\end{example}

\begin{proposition}\label{prop:extensionMeets2}
	Let $F\in \Delta^c_x(\cF)$. Then $F$ is a meet-irreducible in $\extF{x}$ \Iff $F\setminus \{x\}$ is in $\cM^3\cup\cM^4\cup\cM^5\cup\cM^6$.
		\end{proposition} 

\begin{proof}
	Let $F\in \Delta^c_x(\cF)$ be a meet-irreducible in $\extF{x}$. Then $F\setminus \{x\}\in \Delta_x(\cF)$ and has at most one cover in $\Delta_x(\cF)$ otherwise $F$ will have two covers in $\Delta^c_x(\cF)$ and thus is not a meet-irreducible in $\extF{x}$. Suppose $F\setminus \{x\}$ is not in $\cM^3\cup\cM^4\cup\cM^5$ and $F'$ the unique cover of $F\setminus \{x\}$ in $\Delta_x(\cF)$. Since $F$ is a meet-irreducible then its unique cover is $F'\cup \{x\}$ and thus any closed set that contain $F$ contains also $F'\cup \{x\}$. Hence $\Phi(F)=\Phi(F'\cup \{x\})$ and $F\in \cM^6$.

	Conversely, suppose $F\setminus \{x\}$ is in $\cM^3\cup\cM^4\cup\cM^5\cup\cM^6$. We have the following cases:
		
	\begin{enumerate}
	\item $F\setminus \{x\}\in \cM^3\cup \cM^4$: Suppose $F$ has two covers $F_1$ and $F_2$ in $\extF{x}$. Thus $\Phi(F)\subseteq F_1$ and $\Phi(F)\subseteq F_2$ and then $\Phi(F)\subseteq F_1\cap F_2$. Since $F_1\cap F_2 \in \extF{x}$ then $F_1$ and $F_2$ are not covers of $F$, which is a contradiction.

	\item $F\setminus \{x\}\in \cM^5$: Suppose $F$ has two covers $F_1$ and $F_2$ in $\extF{x}$. Without loss of generality suppose $F_1\in \Delta^c_x(\cF)$ and thus $F_2$ contains $\Phi_*(x)$. Then $F_1\setminus \{x\}$ is a cover of $F\setminus \{x\}$ in $\cF$. Moreover $F_2$ is also a cover of $F\setminus \{x\}$ in $\cF$. So $F\setminus \{x\}$ has two covers in $\cF$ which contradicts that is a meet-irreducible.
	
	\item $F\setminus \{x\}\in \cM^6$: Let $F'$ be the unique cover of $F\setminus \{x\}$ in $\Delta_x(\cF)$ with $\Phi(F)=\Phi(F'\cup \{x\})$. Then $F'\cup \{x\}$ is a cover of $F$ in $\Delta^c_x(\cF)$. Since $\Phi(F)=\Phi(F'\cup \{x\})$, then any closed set in $\extF{x}$ that contains $F$ contains also $F'\cup \{x\}$. Thus $F'\cup \{x\}$ is the unique cover of $F$ in $\extF{x}$.
	\end{enumerate}
\end{proof}

\medskip

The following property shows how to compute the set $\cM^6$ from known meet-irreducible sets in $\cF$. 

\begin{property}\label{prop:M6_1}
Let $F$ in $\cM^6$. Then exists $M'\in\cM^3\cup\cM^4$ and $M''\in\cM^2$ such that $F=M'\cap M''$.
\end{property}
\begin{proof}
		Let $F\in \cM^6$ and $F_1$ be its unique cover in $\Delta_x(\cF)$. Then $F$ cannot have two covers containing $\Phi_*(x)$, otherwise their intersection will be $F$ and contains $\Phi_*(x)$, which contradicts that $F $ in $\cM^6$. Let $F_2$ be the cover of $F$ that contains $\Phi_*(x)$. 
	 
		First, since $\Phi_*(x)\nsubseteq F_1$ then there exists $M'$ in $\cM^3\cup\cM^4$ that contains $F_1$. 
		Since $F_2$ contains $\Phi_*(x)$ and it is not a subset of $F_1$, then there exists $M''$ in $\cM^1\cup\cM^2$ such that $M''$ contains $F_2$ and $F_1\nsubseteq M''$. Now suppose that $M''\in \cM^1$. Then $x\in M''$, and, by definition of $\cM^6$, we have $\Phi(F\cup \{x\})=\Phi(F_1\cup \{x\})\subseteq M''$. Thus $F\not = M'\cap M''$ and conclude that $F_2$ must be in $\cM^2$.
	
		Now we show that $F=M'\cap M''$. We have $F\subseteq M'\cap M''$ since $F\subseteq F_1\subseteq M'$ and $F\subseteq F_2\subseteq M''$. Suppose $F\not =M'\cap M''$. Since $F_1\not \subseteq M'\cap M''$ and $\Phi_*(x)\not \subseteq M'\cap M''$ then exists another cover $F_3\in \Delta_x(\cF)$ with $F_3\subseteq M'\cap M''$. This contradicts that $F\in \cM^6$.
	\end{proof}

\medskip

From propositions \ref{prop:extensionMeets1} and \ref{prop:extensionMeets2}, we can characterize exactly the meet-irreducible sets  in $\extF{x}$.

\begin{corollary} \label{meetext}Let $\cF$ be a closure system given by its meet-irreducible sets $\cM$ and $x\in U_\Phi$. Then the meet-irreducible sets of $\extF{x}$ are given by:

$\extM{x}=\cM^1\cup \cM^2\cup \cM^3\cup \{M\cup \{x\} \mid M\in \cup_{j=4}^{6} \cM^j\}$
\end{corollary}

\section{Computing the largest extension }

 Given meet-irreducible sets $\cM$ of a closure system $\cF$ over a set $S$, we propose a polynomial incremental  algorithm to compute meet-irreducible sets $\cM_{max}$ of the largest extension $\cF_{max}$ of $\cF$.  We assume that  $ U_\Phi=\{x_1,x_2,...,x_k\}$ and let  $\extM{i}$ be the set of meet-irreducible sets of  $\extF{i}$, with $\extM{0}=\cM$. So, it suffices to start with $\extM{0}$, and 
 at each step $i$, the algorithm constructs  $\extM{i}$ the meet-irreducible sets  of the extension $\extF{i}$ from   $\extM{i-1}$ in order to obtain  $\cM_{max}=\extM{k}$ the set of meet-irreducible sets of the largest extension.

It is worth noticing that from Corollary \ref{meetext}, we have $\extM{0}\subseteq \extM{1}\subseteq ...\subseteq \extM{k}=\cM_{max}$. Thus the complexity of the algorithm depends on the complexity of the computation of $\extM{i}$ from $\extM{i-1}$. According to Propositions  \ref{prop:extensionMeets1} and \ref{prop:extensionMeets2}, whenever the partition  $\{\cM^1,...,\cM^6\}$ of the meet-irreducible sets of $\extM{i-1}$ is known, then $\extM{i}$ is computed in linear time according to corollary \ref{meetext}


Given the  meet-irreducible sets $\cM$ and $ U_\Phi$, we show how to compute  the partition $\{\cM^1,...,\cM^6\}$ of $\cM$  in polynomial time. Algorithm \ref{partition} follows definitions of the partition given in the previous section. 
   

\begin{algorithm2e}[ht]
\caption{Partition of Meet-irreducible sets  of $\cM$ and $x\in U_\Phi$.  }
\Begin{
	
	\nl \For{$M\in \cM$}{
	    	\If{$\Phi_*(x)\subseteq M$}{
		     \lIf{ $x\in M$}{ Add $M$ to $\cM^1$\;}
		       \lElse {Add $M$ to $\cM^2$\;}
		       }
		\Else{
		    \lIf{ $x\in M^*$}{ Add $M$ to $\cM^3$\;}
		       \Else{
		              \lIf{$\Phi_*(x)\subseteq M^*$}{Add $M$ to $\cM^4$\;}
		              \lElse {Add $M$ to $\cM^5$\;}
		              }
		        }
	     }
	     
	   { Computing $\cM^6$}  
		
	\nl	\For{$M'\in (\cM^3\cup \cM^4)$ and $M''\in \cM^2$}{  
		$F=M'\cap M''$\;
		
	\nl	\lIf{$F$ has a unique cover $F_1$ not containing $\Phi_*(x)$ and $\Phi(F\cup\{x\})= \Phi(F_1\cup\{x\})$}{ Add $F$ to $\cM^6$\;}
		}     
    	}

\label{partition}
\end{algorithm2e}

\begin{theorem}\label{theo:algoExtension}
	Given meet-irreducible sets $\cM$  of a closure system $\cF$ over $S$, the largest extension of the closure system $\cF$ can be computed in incremental-polynomial time. 
\end{theorem}
\begin{proof} We assume that the ground set $S$ has $n$ elements and the number of meet-irreducible sets in $\cM$ is $m$. Given a set $X\subseteq S$, the closure of $S$ can be computed using the closure operator $\Phi(X)=\bigcap\{M\in \cM \mid X\subseteq M\}$ in $O(nm)$. Moreover for $F\in \cF$, closed sets that cover $F$ are given by $Cover(F)=Min_{\subseteq}\{\Phi(F\cup \{a\}), a\in S\setminus F\}$. If $Cover(F)$ is a singleton we denote it by $F^*$. Clearly the set $Cover(F)$ can be computed in $O(n^2m)$. 

The complexity of the loop in Line $1$ of the partition algorithm takes $O(n^2m^2)$. Indeed, for each meet-irreducible $M$, we compute its cover. The complexity of the  loop in line $2$ is bounded by the complexity of the cover computation for each $M'\cap M''$. Thus the complexity of the partition algorithm is in $O(n^2m^3)$.

Finally, for each step $i$, there is a polynomial time algorithm to compute $\extM{i}$ from $\extM{i-1}$. Since the number of meet-irreducible sets is increasing at each step, then the algorithm is incremental polynomial.
\end{proof}
\medskip

The proposed algorithm to enumerate meet-irreducible sets of the largest extension distinguishes at each step $i$  three kinds of meet-irreducible sets: (1) Meet-irreducible sets in $\extM{i-1}$  that remain in $\extM{i}$, (2)  Meet-irreducible sets in $\extM{i-1}$  that are modified  in $\extM{i}$, and (3) New meet-irreducible sets that are created from the duplication of elements in $\cM^3$ and the ones in $\cM^6$. The following question remains open.

\begin{question} Is the number of meet-irreducible sets of the largest extension a polynomial in the size of $\cM$? 	
\end{question}
Notice that if a meet-irreducible $M\in \cM^3$ in the partition of $\extF{i}$ then $M$ cannot be in $\cM^3$ for $\extF{j}$, $j>i$. Indeed, in the second duplication, $M$ will have two covers and it loses the property of being meet-irreducible. Thus answering this question is to bound the number of meet-irreducible sets created by $\cM^6$.

 \section{Conclusion}
We have  given a characterization of an implicational base for the largest extension, and shown that  it can be recognized in polynomial time.
We have also described in this paper a polynomial time  incremental algorithm that builds the set of meet-irreducible sets of the largest extension of a closure system given by its set of meet-irreducible sets.  The question whether the number of meet-irreducible sets of the largest extension is polynomial in the number of meet-irreducible sets of the input closure system remains open. Moreover, the existence of  a polynomial space algorithm that enumerates the set of meet-irreducible sets of the largest extension is more challenging.

\medskip

\paragraph{Acknowledgements} The authors acknowledge the support received from Tassili Project TASSILI $N^{o}$15 MDU944.

\bibliographystyle{plain}
\bibliography{biblio}

\end{document}